\documentclass[11pt]{article}
\usepackage{amsmath,graphicx,amsthm}
\usepackage{hyperref}
\usepackage{xcolor,enumerate}
\usepackage[noend,linesnumbered,ruled,lined]{algorithm2e}
\SetAlFnt{\scriptsize}
\usepackage{comment}
\usepackage{enumitem}
\usepackage{lineno}
\usepackage{soul}
\usepackage{mathrsfs} 

\usepackage{caption}
\usepackage{subcaption}

\newtheorem{theorem}{Theorem}[section]
\newtheorem{lemma}{Lemma}[section]
\newtheorem{observ}{Observation}[section]

\newtheorem{definition}{Definition}[section]

\newtheorem{remark}{Remark}[section]

\providecommand{\keywords}[1]
{
  \small	
  \textbf{\textit{Keywords---}} #1
}

\begin{document}

\title{\bf {Distance-2-Dispersion:\\ Dispersion with Further
Constraints}}
\author{
Tanvir Kaur \footnotemark[1] \and Kaushik Mondal\footnotemark[1]
}

\date{ }
\maketitle
\def\thefootnote{\fnsymbol{footnote}}
\footnotetext[1]{
\noindent Department of Mathematics, Indian Institute of Technology Ropar, India
{\tt tanvir.20maz0001@iitrpr.ac.in, kaushik.mondal@iitropar.ac.in} }


\begin{abstract}
The aim of the dispersion problem is to place a set of $k(\leq n)$ mobile robots in the nodes of an unknown graph consisting of $n$ nodes such that in the final configuration each node contains at most one robot, starting from any arbitrary initial configuration of the robots on the graph. In this work we propose a variant of the dispersion problem where we start with any number of robots, and put an additional constraint that no two adjacent nodes contain robots in the final configuration. That is, the distance between any two nodes with robots must be at least 2. We name this problem as Distance-2-Dispersion, in short, D-2-D. However, even if  the number of robots $k$ is less than $n$, it might be the case that it is not possible for each robot to find a distinct node to reside, maintaining our added constraint. To be more specific, if a maximal independent set is already formed by the nodes which contain a robot each, then other robots, if any, who are searching for a node to seat, will not find one to seat. Hence we allow multiple robots to seat on some nodes only if there is no place to seat. If $k\geq n$, it is guaranteed that the nodes with robots form a maximal independent set of the underlying network.

The graph $G=(V, E)$ has $n$ nodes and $m$ edges, where nodes are anonymous. It is a port labelled graph, i.e., each node $u$ assigns a distinct port number to each of its incident edges from a range $[0,\delta-1]$ where $\delta$ is the degree of the node $u$. The robots have unique ids in the range $[1, L]$, where $L \ge k$. Co-located robots can communicate among themselves. We provide an algorithm that solves D-2-D starting from a rooted configuration (i.e., initially all the robots are co-located) and terminate after $2\Delta(8m-3n+3)$ synchronous rounds using $O(log \Delta)$ memory per robot without using any global knowledge of the graph parameters $m$, $n$ and $\Delta$, the maximum degree of the graph. We also provide $\Omega(m\Delta)$ lower bound on the number of rounds for the D-2-D problem.
\end{abstract}


\section{Introduction}\label{sec:intro}
The aim of the dispersion problem is to place a set of $k(\leq n)$ mobile robots in the nodes of an unknown graph consisting of $n$ nodes such that in the final configuration each node contains at most one robot, starting from any arbitrary initial configuration of the robots on the graph. This problem was introduced in the year 2018 by Augustine et al.\cite{Augustine18}. Later, this problem is studied under various models and with different assumptions in the literature\cite{DasBS21,Molla0M21,MollaM19,MollaMM21,ShintakuSKM20,KshemkalyaniMS22,Barun,Ani,kshemkalyaniMS20,AgarwallaAMKS18}. The main tool used for dispersion is Depth-First-Search traversal and since the robots are equipped with memory, they store the important information required to complete dispersion without getting stuck in a cycle. A natural question arises what will happen if there are some extra constraints imposed on the dispersion problem? As an example, no robot can settle in the one-hop neighborhood of an already settled robot. This led to the generation of Distance-2-Dispersion problem. In this problem, $k$ robots arbitrarily placed on the graph need to attain a configuration such that no two adjacent nodes are occupied by the settled robots. Also, an unsettled robot can settle at node that already contains a settled robot, only if for the unsettled robot there is no other node to settle maintaining the added constraint. With this, there can be $many$ nodes without a settled robot, i.e., no robot to store any information at those nodes, and the graph is a zero storage one, thus the problem becomes interesting if one aims to solve with $less$ memory requirement at each robot.\\
\subsection{Model and the problem} 
Let $G$ be an arbitrary connected undirected graph with $n$ nodes, $m$ edges and maximum degree $\Delta$. The nodes are anonymous, i.e., they have no id. It is a port labelled graph, i.e., each node $u$ assigns a distinct port number to each of its edges from a range $[0,\delta(u)-1]$ where $\delta(u)$ is degree of node $u$. Port numbers that are assigned at the two ends of any edge are independent of each other. Nodes do not have any memory and hence $G$ is a zero storage graph.  

A total of $k$ movable entities are present in the system, which are called robots. Each robot has a unique id from the range $[1,L]$ and each robot knows its id. In some round, if two or more robots are at a single node, we call them co-located and such robots can share information via message passing. \footnote{This is known as the Face-to-Face communication model and has already been considered while studying problems related to mobile robots including exploration~\cite{ShantanuDas19,DereniowskiDKPU15} and dispersion~\cite{Augustine18,KshemkalyaniS21}}Any robot present on some node knows the degree of that node as well as the port-numbers associated with each of the edges corresponding to that node. So, if some robot needs to leave its current node through any particular port number, it can do that. Besides this, whenever any robot moves from a node $u$ to another node $v$, it learns the port number through which it enters the node $v$. 

Our algorithm proceeds in synchronous rounds where in each round robots perform the following steps in order: (i) co-located robots may exchange messages (ii) robots may compute based on available information (iii) robots may move through an edge to some adjacent node from the current node based on its computation in step (ii). We further assume that all the robots start the algorithm at the same time, i.e., from the same round. The time complexity of the algorithm is measured as the number of synchronous rounds required by the robots to complete the task. 
We also study the amount of memory required per robot to run the algorithm.\\\
\
\noindent{\textbf{The problem: Distance-2-Dispersion(D-2-D):}} Given a set of $k\geq 1$ robots placed arbitrarily in a port labelled graph $G$ with $n$ nodes and $m$ edges, the robots need to achieve a configuration by the end of the algorithm where each robot needs to settle at some node satisfying the following two conditions: (i)  no two adjacent nodes can be occupied by settled robots, and (ii) a robot can settle in a node where there is already a settled robot only if no more unoccupied node is present for the robot to settle satisfying condition (i)

The conditions ensure that the distance between any pair of settled robots is at least 2 unless both are settled at the same node. Hence, the nodes with settled robots form an independent set of the graph. And with $enough$ robots, we get a maximal independent set.\\ 
\
\noindent{\textbf{Our contribution:}}
We solve the D-2-D problem for rooted\footnote{The configuration where all the robots are initially placed on a single node of the graph} initial configuration on arbitrary graphs in $2\Delta(8m-3n+3)$ rounds using $O(log \Delta)$ memory per robot in Section \ref{sec:mainalgo}. All the settled robots terminate even without any global knowledge regarding any of the graph parameters $m$, $n$ or $\Delta$.
In Section \ref{sec:lowerbound}, we provide a $\Omega(m\Delta)$ lower bound of the D-2-D problem on the number of rounds considering robots do not have more than $O(log \Delta)$ memory.  Also, if $k\geq n$, it is guaranteed that the nodes with settled robots form a maximal independent set, which can itself be an interesting topic to study in the domain of distributed computing with mobile robots.\\ 
\subsection{Related work} 
Dispersion is the most related problem to our problem as we consider similar model that is considered to solve the dispersion problem. Augustine et al introduced the dispersion problem 
in \cite{Augustine18} for the rooted configuration. They proved the memory requirement by the robots for any deterministic algorithm to achieve dispersion on a graph is $\Omega(\log n)$. The lower bound for any algorithm to perform dispersion on any graph is $\Omega(D)$, where $D$ is the diameter of the graph. 
For rooted graphs, with the knowledge of $m$, $n$, they gave an algorithm that requires $O(\log n)$ memory by the robots to complete dispersion in $O(m)$ rounds\cite{Augustine18}. Kshemkalyani et al.\cite{KshemkalyaniA19} improved the lower bound of running time to $\Omega(k)$ where $k\leq n$. They developed an algorithm for synchronous system which solves dispersion in $O(m)$ rounds using $O(k\log\Delta)$ bits at each robot. However, for an asynchronous system they developed an algorithm which requires $O(max(\log k,\log \Delta))$ bits of memory with the time complexity $O((m-n)k)$. Later Kshemkalyani et al.\cite{KshemkalyaniMS19} significantly improved the result and provided a deterministic algorithm for dispersion in arbitrary graphs in synchronous setting that runs in $O(min(m,k\Delta)\cdot \log l)$ rounds, where $l\leq \frac{k}{2}$, using $O(\log n)$ bits of memory at each robot. Their intuitive idea was to run DFS traversals in parallel to minimize time. The robots required the knowledge of $m$, $n$, $k$ and $\Delta$. Shintaku et al. then presented a dispersion algorithm that does not require such global knowledge \cite{ShintakuSKM20}. Their algorithm solves the dispersion problem on arbitrary graphs in $O(min(m,k\Delta)\cdot \log l)$ rounds using $\Theta(\log(k+\Delta))$ bits of memory at each robot.
Recently, Kshemkalyani et al.\cite{KshemkalyaniS21} provided an algorithm that is optimal in both time and memory in arbitrary anonymous graphs of constant degree. They presented an algorithm which solves dispersion in $O(min(m,k\Delta))$ time with $\Theta(\log(k+\Delta))$ bits at each robot improving the time bound of the best previously known algorithm by $O(\log l)$ where $l\leq \frac{k}{2}$ and matching asymptotically the single-source DFS traversal bounds\cite{KshemkalyaniMS19}. They extend the idea of  \cite{KshemkalyaniMS19} by making the larger size DFS traversal to subsume the smaller size DFS thus avoiding the need of revisiting the nodes of subsumed traversal more than once.

D-2-D, in some sense, is also related to the problem of scattering or uniform distribution. Scattering has been worked mainly for grids\cite{BarriereFBS09} and rings\cite{ElorB11,ShibataMOKM16} though with anonymous robots. Finally, as in some cases, our algorithm forms a maximal independent set, we cite the following study on forming maximal independent set with movable entities, though it is done with stronger model assumptions. Vamshi et al. presented the problem of finding maximal independent set(MIS) using myopic luminous robots\cite{sai} of an arbitrary connected graph where the robots have prior knowledge of $\Delta$, $O(\log \Delta)$ bits of persistent memory and at least 3 hops visibility. Authors also used colors to represent different states and worked under semi-synchronous as well as asynchronous schedulers.\\
\subsection{D-2-D vs dispersion: the challenges}
In the previous works on the dispersion problem, the algorithms use the depth-first search (DFS) traversal with limited memory of the robots \cite{Augustine18,DasBS21,KshemkalyaniS21}. The key idea to achieve dispersion from any rooted configuration is the following. At the starting node, i.e., the root, the robot with the lowest id settles down and the remaining unsettled robots leave the root to visit one of its neighbors. The minimum id robot from this group of unsettled robots settles here. In any round, whenever an unoccupied node is visited by the group of robots for the first time, the minimum id robot settles there. A settled robot on a node represents that the node has already been visited by the group. Using this key idea, DFS traversal for anonymous graphs is feasible by robots since settled nodes help to track cycles and already visited nodes thus ensuring the safe dispersion of the anonymous graph by the robots. For instance, if the group of robots visit a node $v$ from a node $u$ during exploration, and find out that a settled robot is already present at $v$ then as per DFS the group must backtrack to $u$ and explore other unexplored neighbors of $u$ if any. Whenever it finds an empty node, the smallest id robot of the group settles down. The remaining unsettled robots continue the DFS traversal.
However, to do the backtracking successfully, the robots need to remember the path they have already explored and this needs high memory requirement. To avoid this, the algorithm instructs each settled robot to store the information required to backtrack from the respective node where the robot is settled. Observe that, the group never backtracks from a node where no robot is settled. To be more specific, each settled robot remembers the parent pointer as the port number it used while entering into the node it is settled at. The group of backtracking robots can use this parent pointer and can successfully backtrack from this node. Thus dispersion can be achieved with $\textbf{O}(\log \Delta)$ memory per robot. 

In D-2-D, the main motive is to achieve dispersion such that no two adjacent nodes have settled robots. In other words, the distance between any two settled robots must be at least two. To do so, we may face the following challenges.
 \begin{itemize} 
 \item Since there is no settled robot present in the one-hop neighborhood of any settled robot, the information regarding the parent pointer of those neighboring nodes is difficult to be stored and subsequently used while backtracking. This may lead to high memory requirements. 
         
         \item The robots can settle at a node $u$ if and only if there are no settled robots at any of the neighbors of $u$. As the maximum degree $\Delta$ can be large, this may lead to high time complexity.
\end{itemize}

\section{Warm-Up: D-2-D with $O(\Delta\log\Delta)$ Memory per Robot}\label{sec:warmup}
In this section, we provide an informal discussion on a straightforward solution of the rooted D-2-D without bothering about the memory requirement per robot or the time complexity. Our algorithm is based on the depth-first search traversal (albeit with some modification) that solves the dispersion problem as we discussed above in Section \ref{sec:intro}. Later in Section \ref{sec:mainalgo}, we improve over this solution. 

While solving the dispersion problem, encountering a settled robot while doing forward exploration, implies the presence of cycle to the moving group of unsettled robots and then backtracking is done with the help of stored parent pointers at the settled robots. We do the following modification to solve our problem. As the group may need to backtrack from an unoccupied node in our D-2-D problem, it is required to store the parent pointers of the unoccupied nodes too. Note that all the neighboring nodes of any occupied node must be unoccupied. We instruct each settled robot to remember the parent ports of all its neighbors including itself. Basically each settled robots work as $virtually settled$ robot at its neighbors. The notion of virtually settled means that although there is no settled robot present at that node, yet no robot from the visiting group can occupy it. However, to store as well as to provide the stored parent pointers, each settled robot must meet the moving group of robots whenever the group reaches one of its neighboring nodes.

To achieve this, each settle robot does a back and forth movement from its position to its neighboring nodes. To be more specific, let a robot $r$ settles in round $T$ at a node $u$ of degree $\delta$. It visits $u(0)$ in round $T+1$, comes back to $u$ in round $T+2$, visits $u(1)$ in round $T+3$, and so on. It visits $u(0)$ again after $u(\delta-1)$ is visited.  The settled robot $r$ stops only when it meets the group of unsettled robots at some node, say $v=u(p)$, and at some round say $T'$. It stays with the group at $u(p)$ till the group leaves $u(p)$, say in round $T''$. Then $r$ comes back to $u$ in round $T''+1$ and again starts visiting its neighbors one by one as described earlier. 

Note that only this does not solve the problem as the group of moving robots may reach a neighboring node of some occupied node but at that time the respective settled robot may visit its another neighbor. To solve this issue, the algorithm instructs the moving group to wait for $2\Delta$ rounds at each node $v$ which ensures that the moving settled robot must meet the group within this time period. For simplicity, let us assume that all robots know $\Delta$.\footnote{In the next section we  remove this assumption in our main algorithm using the idea that whenever the group of unsettled robots visits a new node, each unsettled robot updates the value of maximum degree seen by itself till now, as the value of $\Delta$.} If $v$ is occupied, then the settled robot must meet the group within 2 rounds; else if $v$ is a neighbor of an occupied node, then the settled robot that is working as a virtually settled robot must meet the group within the $2\Delta$ waiting time. 

Now we provide the algorithm. If the group is in forward exploration phase, the following are the possibilities.
\begin{itemize}
\item The group meets at least one virtually settled robot and finds that none of the virtually settled robots who meet the group has the parent pointer for this node, the group understands that it is visiting this node for the first time and continues the DFS traversal in the forward exploration phase after providing the parent pointer to each of the virtually settled robots. This is possible as all the settled robots that meet this group wait with this group till the group leaves. 
\item The group meets at least one virtually settled robot and finds that at least one of the virtually settled robots comes with the parent pointer for this node, the group understands that this node is already explored earlier, and subsequently backtracks. 
\item The group finds that the node is occupied, it goes to $backtrack$ phase, and backtracks with the help of the parent pointer stored at the settled robot.
\item The group sees no settled or virtually settled robots, then the minimum id robot from the group settles there.
\end{itemize}
If the group is in backtracking phase, it must meet at least one virtually settled robot or a settled robot at the node. The group checks if all the ports associated to this node is already explored or not, by looking at the parent pointer and the port though which it just backtracks to this node. The following are the decisions. 
\begin{itemize}
\item  If all ports are explored, it continues the backtracking.
\item If some port remains to be explored, it changes to forward exploration phase. 
\end{itemize}
Figure \ref{fig:fig} represents the implementation of the algorithm on an example.

\begin{figure}[ht!]
\begin{subfigure}{.5\textwidth}
  \centering
  \includegraphics[width=.6\linewidth]{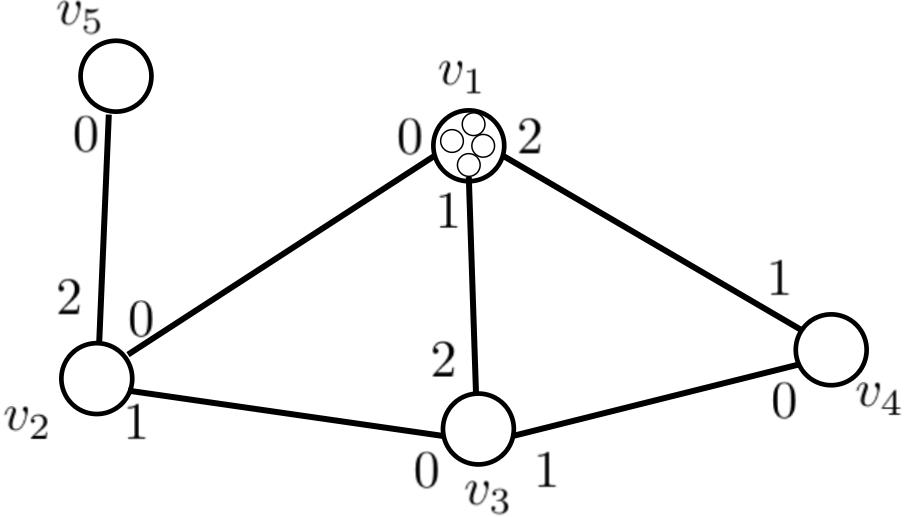}  
  \caption{}
  \label{fig:sub-first}
\end{subfigure}
\begin{subfigure}{.5\textwidth}
  \centering
  \includegraphics[width=.6\linewidth]{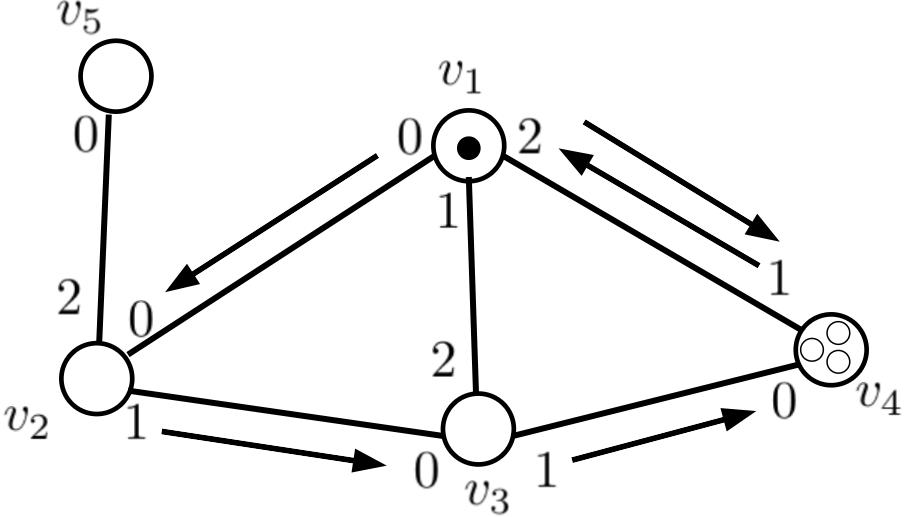}  
  \caption{}
  \label{fig:sub-second}
\end{subfigure}

\begin{subfigure}{.5\textwidth}
  \centering
  \includegraphics[width=.6\linewidth]{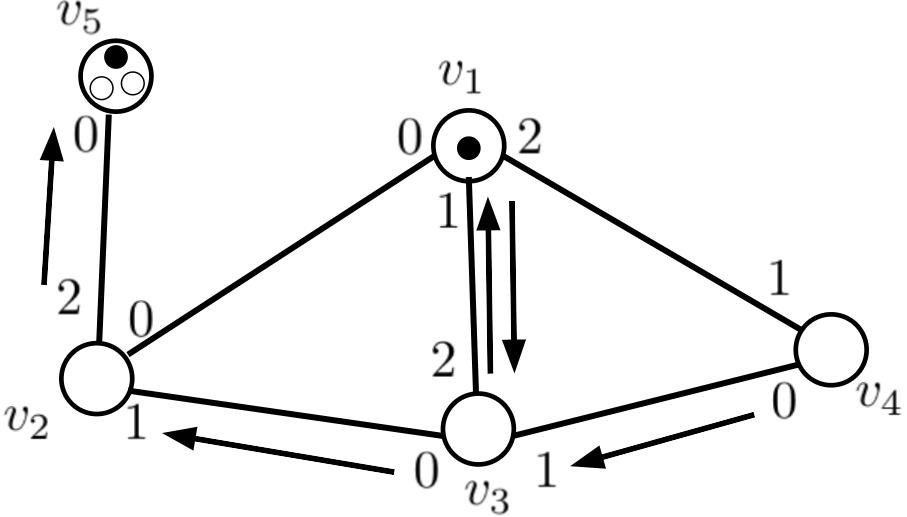}  
  \caption{}
  \label{fig:sub-third}
\end{subfigure}
\begin{subfigure}{.5\textwidth}
  \centering
  \includegraphics[width=.6\linewidth]{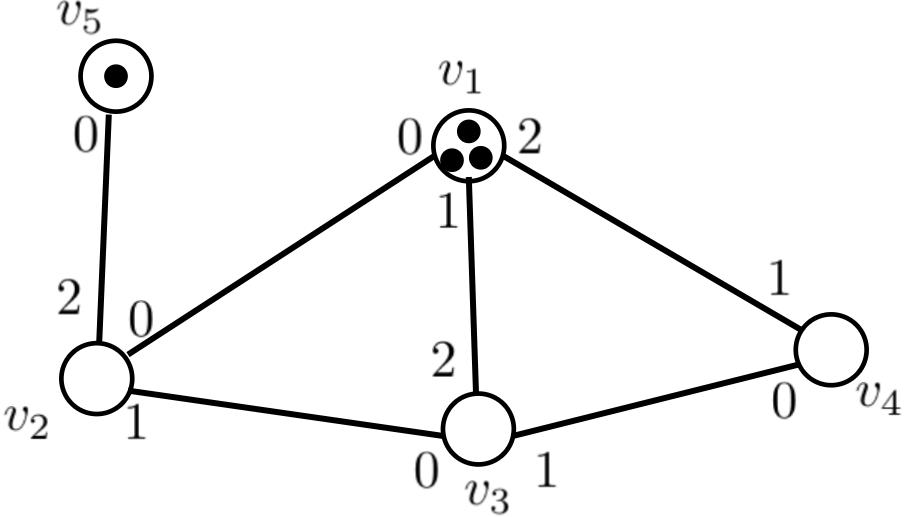}  
  \caption{}
  \label{fig:sub-fourth}
\end{subfigure}
\caption{(a) The initial configuration with four robots at $v_1$.
(b) One robot settles at $v_1$ while the remaining unsettled robots do the DFS traversal through $v_2$, $v_3$, $v_4$. None of them settles at any of these nodes as the robot settled at $v_1$ works as the virtually settled robot at these nodes. The settled robot at $v_1$ maintains the parent pointer for the group at all its neighboring nodes. From $v_1$ the robots backtrack due to the presence of a settled robot on that node. (c) The unsettled robots continue the DFS traversal and when the group reaches $v_5$, a robot settles there. (d) The remaining robots continue the traversal, finish the traversal at $v_1$, then do it again and finally achieve the desired configuration.
}
\label{fig:fig}
\end{figure}

To store the parent pointer for each of the neighbors, each settled robot requires $O(\Delta\log \Delta)$ memory i.e., $\log \Delta$ memory per port number. This leads to high memory requirements and waiting for $2\Delta$ rounds for the moving group at each newly explored node leads to a high run time of the algorithm. Both of these issues are mentioned when we discussed about the challenges in Section \ref{sec:intro}. The run time of this algorithm becomes $2\Delta(4m-2n+2)$. This is because we just run the DFS traversal for dispersion that takes $4m-2n+2$ time as done in \cite{Augustine18,KshemkalyaniMS19} assuming all the unoccupied neighbors of each occupied node are virtually occupied. Also, we keep all the necessary information with the settled robots. After completing the DFS traversal once, if some robots are left unsettled then they settles at the root node. Note that as the settled robots keep moving, it is desirable that the robots can terminate. However, here we do not discuss about the termination of the algorithm, but it can be achieved in a similar way we do in Section \ref{sec:mainalgo} in our main algorithm.

\section{D-2-D with Termination: $O(\log \Delta)$ Memory per Robot}\label{sec:mainalgo}
In this section, we present an algorithm that improves the memory requirement of the algorithm discussed in section \ref{sec:warmup} as well as we include termination of the robots without any global knowledge. First we provide a high level idea. The idea is not to do the usual DFS traversal and in turn, no settled robot needs to maintain parent pointers for all its unoccupied neighbors. Each settled robot stores the parent pointer corresponding to only one of its neighboring nodes along with the parent pointer corresponding to the node where it is settled. This bounds the path from the last settled robot to the node where the unsettled group is currently present. If the length of this path is high, then in the worst case, all the nodes in this path may remain unoccupied as all of them can be the neighbor of some particular occupied node $v$, and in this case the robot that is settled at $v$, needs to remember all the parent pointer of the unoccupied nodes present in that path. We restrict the depth of the traversal during forward exploration by 2 from the node where a robot settled last. Now, the question is how does the information corresponding to only one neighbor suffice? Similar to the previous algorithm, each settled robot visits each of its neighbors one by one in the subsequent rounds once it is settled while the group of unsettled robots waits for $2\phi$ rounds, where $\phi$ is the maximum degree observed by the group of unsettled robots till reaching the current node, in order to check the presence of any settled robot in its one-hop neighborhood. Similar to the algorithm described in the Section \ref{sec:warmup}, the group of unsettled robots update the value of $\phi$ when they reach a new node $u$ by $max\{\phi, \delta_u\}$ as they have no prior knowledge of the maximum degree of the graph. The group of unsettled robots maintains a counter $r_i.dist$ from the last encountered settled robot which stores the distance from the settled robot. The variable $r_i.dist$ can take values 0, 1 or 2. The group of unsettled robots decides to do forward exploration or backtrack based on the value of $r_i.dist$, the number of settled robots visiting this group during the waiting period and the $id$s of these visiting settled robots. Thus, the state of any robot can be $explore$ or $backtrack$. The robots maintain another variable $r_i.stage$ which takes the value $2$ when the D-2-D configuration is achieved and the termination stage has started. In order to terminate the algorithm, the settled robots maintain variables $r_i.count$, $r_i.count'$, and $r_i.act\_settled$ whose description is given in the table \ref{table:symbols}. 

Whenever the group of unsettled robots moves from a node $u$ containing a settled robot $r_j$ with the lowest id to one of its neighbors, say $u(p)$ i.e. the unsettled group of robots have $r_i.dist=1$ at $u(p)$, the state is $explore$ and no other settled robot visits the group at $u(p)$ except $r_j$, then the parent pointer of $u(p)$ is required to be stored by the settled robot $r_j$ ($r_j.virtualparent$ as defined in Table \ref{table:symbols})  and the group explores further. From the settled robot's point of view, it stores its parent pointer ($r_j.parent$ as defined in Table \ref{table:symbols}) along with the parent pointer of only one neighbor $u(p)$. The idea is $r_j$ may store the parent pointer corresponding to another neighbor $u(q)$, say, when q is the last port that was explored from the node where $r_j$ is settled.
Else, if one or more settled robots with id lower than $r_j.id$ meet the group at $u(p)$, then the group backtracks to the node $u$ along with the settled robot $r_j$ and $r_j$ does not need to store the parent pointer of $u(p)$ in this case. All the settled robots including $r_j$ which meet the group of unsettled robots will increment their value of $count$ variable each time. Note that, to avoid needless increments in the value of the $count$ variable by the settled robots, when the settled robots encounter the group of unsettled robots at a node during their back and forth movement, they wait with the group until it leaves that node.

When $r_i.dist=2$ as the group reaches a node say $v$ after exiting $u(p)$ where $r_j$ is settled, a robot settles there only if there are no neighboring settled robots. In this case, the newly settled robot keeps the parent pointer as the group moves forward. Else, if at the node $v$, one or more settled robots are visiting, then the group backtracks as there are options to explore this node later through some other already settled robots with possibly lower id than $r_i$. For instance, if the unsettled group of robots at $dist=2$ from the node $u$ are witnessing their waiting period of $2\phi$ rounds and during this waiting period, a robot $r_m$ visits. The group not only understands that no robot can settle at this node, but also, this node can be explored from the node that has the settled robot $r_m$. In other words, if the group of unsettled robots are at $dist=2$ from a node with a settled robot and are in the $explore$ state, and some settled robot $r_m$ visits the group during the waiting period then this implies a direct link to this node from $r_m$. This means the current node can be explored from the node where $r_m$ is settled. So, in no circumstances, any settled robot, here $r_j$, is keeping parent pointers of a node that is at least two hops away. All the settled robots $r_j$ which meet the group of unsettled robots increment the value of $r_j.count$ each time they meet the group of unsettled robots. These set of steps comprise the stage 1 of the algorithm.
After the last robot settles, it begins the next stage, namely, the termination stage. The last settled robot starts from the root node and follows similar path as described in the above paragraph. In this termination stage, the last robot $r_L$ acts as the group of unsettled robots in the stage 1 while the robot $r_j$ which originally settled during the Distance-2-Dispersion, set $r_j.act\_settled=1$ in the same order as it settled in the previous stage. Thus, the termination stage is a replica of stage 1. Each time the $act\_settled$ robot $r_j$ meets the last settled robot $r_L$ during its traversal, $r_j$ increments the value of $r_j.count'$. As and when the value of $r_j.count$ becomes equal to $r_j.count'$, the robot $r_j$ terminates. In this way the robots terminate and Distance-2-Dispersion with termination is achieved.

\begin{center}
\begin{table}[ht!]
\begin{tabular}{|l|l|} 
  \hline
  Variables & Descriptions \\ 
  \hline
 $r_i.parent$ & This variable contains the parent port of the node $u$ \\ & where $r_i$ is settled in stage 1 or $act\_settled$ in stage 2.\\ & Else, $r_i.parent=-1$. \\  \hline
  $r_i.portentered$ & The port through which robot $r_i$ enters the current\\ & node. Initially $r_i.portentered=-1$ for all the robots.\\
  \hline
  $r_i.virtualparent$ & The parent port of $u(p)$ where $p$ be the last port that\\ & was explored from the node where $r_i$ is settled.\\ & Initially $r_i.virtualparent=-1$  for all the robots.\\ \hline
   $r_i.dist$ & Each unsettled robot $r_i$ maintains the distance from\\ & the settled robot it last encountered during its \\& traversal. According to our algorithm, $r_i.dist\leq 2$. For \\& each settled robot, $r_i.dist= 0$ if it is at the node \\& where it is settled, else $r_i.dist= 1$.  \\ 
   \hline
    $r_i.special$ & A robot $r_i$ settled at some node $u$, say, updates\\& $r_i.special=1$ only when the group of unsettled \\& robots is at $u$ with $r_i$ and will move through one \\& of the adjacent edges of $u$ in the $explore$ state.  \\& For other settled robots, $r_i.special=0$ and for any \\& unsettled robot, $r_i.special=-1$.\\ 
    \hline
    
    $r_i.settled$ & Takes value 1 if $r_i$ is a settled robot in stage 1, else\\& takes 0.  \\ 
    \hline
    $r_i.count$ & If $r_i$ is a settled robot in stage 1, this variable\\& counts the number of times $r_i$ meets the group \\& of unsettled robots. \\
    \hline
    $r_i.stage$ & This variable can take values $1$ or $2$ where\\& $r_i.stage=1$ indicates stage 1 of the algorithm\\&  whereas $r_i.stage=2$ indicates the stage 2.\\
    \hline
    $r_i.act\_settled$ & Takes value $1$ if $r_i$ settles in the stage 2, else\\& takes value $0$.\\
    \hline
    $r_i.count'$ & In the stage 2, if $r_i$ is an $act\_settled$ robot, this \\& variable counts the number of times $r_i$ meets\\&  the robot $r_L$ with $r_L.terminate=1$.\\
    \hline
\end{tabular}
\caption{Description of variables}\label{table:symbols}
\end{table}
\end{center}

\subsection{The Algorithm}\label{sec:3.1}
In this section, we describe our algorithm that can be run by the robots with $O(log \Delta)$ memory per robot to achieve rooted D-2-D and terminate in $2\Delta(8m-3n+3)$ rounds.

Consider an arbitrary graph $G$ and let $k$ robots be initially placed on a single node, say $u$, of $G$. Each robot $r_i$ maintains variables $r_i.parent$, $r_i.settled$, $r_i.dist$, $r_i.special$, $r_i.portentered$, $r_i.virtualparent$, $r_i.count$, $r_i.stage$, $r_i.act\_settled$ and $r_i.count'$. We have defined the variables in detail in Table \ref{table:symbols}. Apart from that, each unsettled robot maintains its state, which can be either $explore$ or $backtrack$. In $explore$ state, it does forward exploration by moving through a computed port number from the current node, whereas in $backtrack$ state, it learns the parent pointer from a settled robot and backtracks through that port.
Robots also maintain a variable $\phi$ which is initialized to $\delta(u)$, where $u$ is the root. Each unsettle robot updates $\phi$ as and when they see some node with degree more than the current value of $\phi$. However, settled robots do not modify $\phi$ once they are settled. Note that the each robot in the group of unsettled robots has the same value of $\phi$; also $\phi\leq\Delta$, the maximum degree of the graph. More specifically, when the group of unsettled robots reach a new node, say $v$, each robot in the group updates the value of $\phi$ with $max\{\delta(v), \phi\}$. Our algorithm works in two stages, stage 1 and stage 2. In stage 1, robots achieve D-2-D and in stage 2, robots terminate. The group of unsettled robots run the algorithm in phases, where each phase consists of $2\phi$ rounds. Since each unsettled robot agrees on their $\phi$ value, the group starts and ends each phase at the same round.

Our algorithm starts by settling the minimum id unsettled robot, say $r_1$, at the root node $u$ at the end of phase 1. The settled robot $r_1$ updates  $r_1.parent \leftarrow -1$, $r_1.settled\leftarrow 1$ and $r_1.dist\leftarrow 0$. The remaining unsettled robots update their $state$ to $explore$ and since they are present with the settled robot $r_1$, the variable $r_1.special$ is updated to $1$. The unsettled robots update the port number as $r_i.portentered=(r_i.portentered+1) mod \delta$, where $\delta$ is the degree of the node $u$. The unsettled robots move along the incremented port number, which in this case will be port 0, to the neighboring node $u(0)$, and update $r_i.dist=1$. With this, the phase 1 ends for the unsettled robots. After reaching this node, the unsettled robots update $\phi$ and wait for $2\phi$ rounds. As this is a rooted configuration and $r_1$ is the first robot to settle(which is currently present along with the group of unsettled robots), no other settled robot visits this group during the wait of $2\phi$ rounds. This implies that the current node can be further explored if there are unexplored edges, i.e., degree of $u(0)$ is more than one.  The settled robot $r_1$ updates $r_1.count \leftarrow 1$. Let $p_1$ be the port through which robots entered $u(0)$. The settled robot $r_1$ updates $r_1.virtualparent\leftarrow p_1$, return through port $p_1$ to the root $u$ and updates $r_1.special=0$.
The unsettled robots update $r_i.portentered=(p_1+1)mod (\delta(u(0))$), where $\delta(u(0))$ is the degree of node $u(0)$. The unsettled robots move through the updated value of $portentered$ in the $explore$ mode and update $r_i.dist=2$ if $r_i.portentered\neq r_1.virtualparent$. Else, the unsettled robots move through the updated value of $portentered$ in the $backtrack$ mode and update $r_i.dist=0$. This ends phase 2 for the group of unsettled robots. Basically, the group of unsettled robots move only at the end of each phase.

After the group of unsettled robots reaches a new node at distance 2 apart from $u$, they update $\phi$ and wait for $2\phi$ rounds while the settled robot $r_1$ now has $r_1.special=0$, and it continuously traverses through each of the ports of $u$ one by one. During this waiting period, if any settled robot $r_j$ visits the group of unsettled robots, they backtrack to the previous node using the $r_i.portentered$. However, the settled robot $r_j$ which meets this group of unsettled robots, waits with the group of unsettled robots unless the group leaves that node and increments the value of $r_j.count$. However, if no settled robot visits this group then the minimum id robot from the group of unsettled robot settles at this node. 

So each phase of unsettled robots corresponds to one edge traversal and each phase requires $2\phi$ rounds. Settled robots do not bother with phases, they either wait with the group of unsettled robots or continue back and forth traversal. So, the time complexity of stage 1 of our algorithm depends on how many phases the group of unsettled robots works before the last robot, i.e., the largest id robot, say $r_L$, settles. The settled and unsettled robots decide what to do based on the following cases. First we write what settled robots do.

\begin{itemize}
    \item If a settled robot $r_j$ has the value $special=0$, then it continues its visit to each of its neighbors one by one and keeps modifying $r_i.dist$ accordingly. As and when the settled robot meets the group of unsettled robots, it increments the value of $r_j.count$ and waits with the group unless it leaves that node. Post that, the settled robot resumes its visit to each of its neighbors one by one.
    \item If a settled robot $r_j$ has the value $special=1$, then it moves along with the group of unsettled robots in the next round to the neighboring node and waits along with the group of unsettled robots. The settled robot $r_j$ increments the value of $r_j.count$. During this waiting period, if any settled robot with id lower than $r_j.id$ visits then $r_j$ does not store the parent pointer for this node and return back to its original position after the group leaves. However, if no settled robot with id smaller than $r_j.id$ visits the group of unsettled robots then the settled robot $r_j$ updates $r_j.virtualparent=r_j.portentered$.
\end{itemize}
The pseudo-code for the settled robots is given in Algorithm \ref{alg:settle}.
\begin{algorithm}[ht!]
\caption{Algorithm for each settled robot $r_j$}\label{alg:settle}
\If{$r_j.id$ is the largest among the group from which it settled}
    {
    set $r_j.terminate=1$
    }
\Else
{
\If{$r_j.special=0$}{
do the back and forth movement through all its ports and updates $r_j.dist$ value to $0$ or $1$ according to its position. \\
\If {the settled robot $r_j$ meets the group of unsettled robots}{
{wait with the group of unsettled robots till the group leaves that node}\\
{increment the value of $r_j.count$}\\
{move to the original position where it was settled, if not already there, and resume the back and forth movement}}
\ElseIf{$r_j$ meets any robot $r_i$ with $r_i.stage = 2$}
{
{move to its original position where it was settled, if not already there, and wait until it meets the robot $r_L$ with $r_L.terminate=1$}\\
    {do not increment the value of $r_j.count$}\\
    {set $r_j.stage=2$}\\
    \If{the robot $r_L$ with $r_L.terminate=1$ visits this node where $r_j$ is waiting}
    {
    set $r_j.act\_settled=1$, $r_j.parent=r_L.portentered$, $r_j.special=1$
    }
    }  

}
\If {$r_j.special=1$}{
{move along with the group of unsettled robots through updated value of $r_i.portentered$ and wait till the group leaves.}\\
{increment the value of $r_j.count$}\\
\If{no settled robot visits with id lower than $r_j.id$}{
{set $r_j.virtualparent=r_i.portentered$}\\
{move to the original position via $r_j.portentered$}\\
}
\Else{
{move to the original position via $r_j.portentered$}
}
{update $r_j.special=0$}
}
}
\end{algorithm}

Now we see how the unsettled robots work. When the group of unsettled robots reach a node at the end of some phase, then the decision of what to do at the end of the next phase is made based on the following cases:

\begin{itemize}
    \item When $r_i.dist=1$ and $state=explore$ at the beginning of a phase
    \begin{itemize}
        \item If no settled robot with $special=0$ visits the group of unsettled robots such that the visiting settled robot's id is lesser than the id of $r_j$ with $r_j.special=1$, then the group updates $r_i.portentered=(r_i.portentered+1) mod \delta$. Note that whenever the group of unsettled robots have $r_i.dist=1$ and $state=explore$, there is definitely a settled robot $r_j$ present with the group having $r_j.special=1$. If $r_i.portentered=r_j.virtualparent$, then the group of unsettled robots backtracks to the previous node. Else if $r_i.portentered\neq r_j.virtualparent$ then the group of unsettled robots leaves the current node via the updated port number in the $explore$ state. 
        
        \item If at least one settled robot, with $special=0$ and having id lesser than id of $r_j$ with $r_j.special=1$, visits the group of unsettled robots, then the group backtracks via $r_i.portentered$. 
        
            \end{itemize}
    \item When $r_i.dist=2$ and $state=explore$ at the beginning of a phase 
    \begin{itemize}
        \item If no settled robot visits the group of unsettled robots in $2\phi$ rounds, then the lowest id robot from the group settles on this node. 
        \item If one or more than one settled robot visits the group of unsettled robots, then the group backtracks via the $port\_entered$. 
        
    \end{itemize}
\end{itemize}

However, when the group of unsettled robots is in $state=backtrack$, the group of unsettled robots reaches a node which has been visited earlier. Thus, the node which is visited in $backtrack$ state has either a settled robot on it or any settled robot in its one hop neighbor has stored the virtual parent for it. After backtracking to the current node, the group of unsettled robots decides to further explore or backtrack based on the following cases.
\begin{itemize}
\item When $r_i.state=backtrack$ at the beginning of a phase
    \begin{itemize}
         \item If all the ports are already explored i.e. $(r_i.portentered+1)mod \delta$ is equal to the $r_m.parentpointer$ where $r_m$ is the settled robot on that node or $r_j.virtualparent$ where $r_j$ is the settled robot which has stored the parent pointer for its neighbor. Recall that this information can be exchanged during the waiting time. In this case, the group of unsettled robots backtracks through $r_m.parentpointer$ or $r_j.virtualparent$ of the node.
         \item If the $(r_i.portentered+1)mod \delta$ is not equal to the $r_i.parentpointer$ of the node then the unsettled group of robots change their $state$ to $explore$ and move through $r_i.portentered$.
    \end{itemize}
\end{itemize}

The pseudo-code for the unsettled robots is given in Algorithm \ref{alg:cap}.
\begin{algorithm}
\caption{Algorithm for each unsettled robot $r_i$}\label{alg:cap}
{initialise $r_i.portentered=-1$, $r_i.dist=0$, $r_i.state=explore$, $r_i.special=-1$, $r_i.settled=0$, $r_i.virtualparent=-1, \phi=0$, $r_i.terminate=0$, $r_i.stage=1$}\\
\For{$phase=0$}
{
  {the minimum id robot $r_j$ settles on the node after waiting for $2\phi$ rounds}\\
{set $r_j.parent=-1$}\\
{set $\phi= \delta(u)$}\\
{$r_i.portentered=(r_i.portentered+1)$mod$\delta$}\\
{move through $r_i.portentered$}\\
{$r_i.dist$=$r_i.dist+1$}\\
}
\For{$phase > 0$}
{
\If{$r_i.dist=0$}
{
    {$r_i.dist = r_i.dist +1$}   
}
\If{$r_i.dist=1$ and $r_i.state=explore$}{
{wait for $2\phi$ rounds}\\
    \If{any settled robot visits with $r_j.special=0$ and $r_j.id$ is smaller than the settled robot already present with the group}{
    {$r_i.state=backtrack$}\\
    {update $\phi$ = $max\{\phi , \delta(u)\}$}\\
    {move through $r_i.portentered$}\\
    {decrement $r_i.dist$ value to 0}\\
    }
    \Else{
    {$r_i.state=explore$}\\
    {update $\phi$ = $max\{\phi , \delta(u)\}$}\\
    {$r_i.portentered=(r_i.portentered+1)$mod$\delta$}\\
        \If{$r_i.portentered=r_j.virtualparent$}{
        {$r_i.state=backtrack$}\\
        {move through $r_i.portentered$}
        }
        \Else{
         {move through $r_i.portentered$}\\
         {increment $r_i.dist$ to 2}\\
        }
   
    }
  }
\If{$r_i.dist=2$ and $r_i.state=explore$}{
{wait for $2\phi$ rounds}\\
{update $\phi$ = $max\{\phi , \delta(u)\}$}\\
    \If{no settled robot visit the group of unsettled robots}{
    {the robot $r_i$ with lowest id settles on that node and sets $r_i.parent=r_i.portentered$ and $r_i.special=1$}\\
    {each unsettled robot makes $r_i.portentered=(r_i.portentered+1)$mod$\delta$} and
    {$r_i.dist=0$}\\
        \If{$r_i.portentered$ = parent pointer of the settled robot at the current node}
        {
            $r_i.state=backtrack$
        }
    {move through $r_i.portentered$}\\
    }
     \Else{
     {$r_i.state=backtrack$}\\
     {move through $r_i.portentered$}\\
    }
} 
\If{$r_i.state=backtrack$}{
{wait for $2\phi$ rounds}\\
{decrement $r_i.dist$ by 1}\\
{$r_i.portentered=(r_i.portentered+1)$mod$\delta$}\\

    \If{the settled robot $r_j$ has $r_j.parent=-1$}{
        \If{$r_i.portentered=0$}{
        the unsettled robots settle at the root }
        \Else{
        {$r_i.state=explore$}\\
        {move through $r_i.portentered$}
        }
        }
    \Else{
        \If{$portentered=parent$ or $virtualparent$}{
        {$r_i.state=backtrack$}\\
        {move through $portentered$}\\
        }
        
    \Else{
    {$r_i.state=explore$}\\
    {move through $portentered$}\\
    {increment $r_i.dist$ by 1}\\
    }
}
}
}
\end{algorithm}

Using this we achieve D-2-D configuration and stage 1 completes. Now the stage 2 begins. The robot with the largest id, say $r_L$, is said to be the last settled robot and after it settles it sets $r_L.stage=2$. In order to terminate the algorithm, the $r_L$ restarts similar traversal as described in stage 1 from the root. Now two cases arise. Either the robot $r_L$ settles at a node other than the root node (in this case it goes to the root node as described below) or $r_L$ settles at the root itself. In any case, while at the root, $r_L$ sets $r_L.terminate=1$. From root, $r_L$ starts mimicking the stage 1 again by acting as the group of unsettled robots, to help the remaining settled robots terminate. Finally, $r_L$ terminates after reaching the node where it settled at the end of stage 1. Note that, if and when $r_L$ backtracks till the root, as some settled robots may meet $r_L$ on the path since they are continuing their back and forth movement, yet no settled robot increments their respective count variable since $r_L$ is not an unsettled robot anymore.
\begin{itemize}
    \item When the last robot $r_L$ settles at node $u_l$ other than the root
    \begin{itemize}
        \item The robot $r_L$ backtracks through $r_L.portentered$. It continues backtracking through the parent pointer to reach the root node. After reaching the root node, the robot $r_L$ begins mimicking stage 1 as the group of unsettled robots.
    \end{itemize}
    \item When the remaining robots settle at the root and $r_L$ is the largest id robot among the group
    \begin{itemize}
        \item The robot $r_L$ begins mimicking stage 1 as the group of unsettled robots.
    \end{itemize}
\end{itemize}

When a robot $r_i$ with $r.settled=1$ meets $r_L$ with $r_L.terminate=1$ or any robot $r_j$ with $r_j.stage=2$, it understands stage 2 is under progres, updates $r_i.stage=2$, goes back to its original position and waits for $r_L$ with $r_L.terminate=1$ to arrive there. When a robot $r_j$ other than $r_L$ meets $r_L$ with $r_L.terminate=1$ for the first time, following two cases are possible. If this meeting is done at the node where $r_j$ settled in phase 1, it understands it has to wait here till $r_L$ is here by acting as $r_j$ is also an unsettled robot which is with $r_L$. Then $r_j$ updates $r_j.act\_settled=1$ as if it becomes settle by the end of this phase and start following the algorithm of a settled robot of stage 1. The only difference is that, now it increments $r_j.count'$ instead of $r_j.count$. Else if this meeting is done at a neighboring node of the node where $r_j$ settled in phase 1, it changes $r_j.stage=2$, goes back to its original position and waits for $r_L$ with $r_L.terminate=1$ to arrive there.

The above description shows that, any settled robot $r$ in phase 2 increments their variable $r.count'$ only after becoming $act\_settled$, i.e., only after the corresponding round of stage 1 when $r$ became settled and started incrementing its $r.count$. Since $r_L$ works as the unsettled robot by following the algorithm of unsettled robots in stage 1 and its id is the largest, $r_L$ will settle in phase 2 again at the end and will terminate. By that time, all settled robots $count$ value match with their $count'$ value, and all terminate. With this, below we provide the algorithm of stage 2 in a formal way.

When the robot $r_L$ with $r_L.terminate=1$ and $r_L.stage=2$ reaches the root $u$, it initializes the value of $\phi=\delta(u)$, and begins the traversal in order to terminate the settled robots. Similar to the algorithm described above for the unsettled group of robots in stage 1, every time $r_L$ reaches a node, say, $v$, it waits for $2\phi$ rounds at that node. Also, it updates the value of $\phi$ with $max\{\phi, \delta(v)\}$.
After reaching a new node, decisions are made based on the following cases:
\begin{itemize}
    \item When $r_L.dist=1$ and $r_L.state=explore$
        \begin{itemize}
            \item If any settled robot visits with $r_j.act\_settled=1$, $r_j.special=0$ and $r_j.id$ is smaller than the id of $act\_settled$ robot already present with $r_L$ then $r_L$ backtracks via $r_L.portentered$
            \item if any settled robot visits with $r_j.act\_settled=1$ and $r_j.special=0$ but $r_j.id > $ the id of $act\_settled$ robot already present with $r_L$ or no $act\_settled=1$ robot visits then $r_L$ is set to explore and it moves through the incremented value of $portentered$
        \end{itemize}
    \item When $r_L.dist=2$ and $r_L.state=explore$
        \vspace{-0.05cm}
        \begin{itemize}
            \item If no settled robot with $r_j.act\_settled=1$ visits then this indicates that it is the original position where $r_L$ is settled and hence set $r_L.act\_settled=1$. Terminate $r_L$.
            \item if any settled robot visits with $r_j.act\_settled=1$ then $r_L$ backtracks through $r_L.portentered$
            \item if there is a robot $r_j$ present at the node with $r_j.stage=2$ but $r_j.act\_settled=0$ then $r_j.parent$ is set to $r_L.portentered$ and $r_j.special=1$. 
        \end{itemize}
    \item When $r_L.state=backtrack$
        \begin{itemize}
            \item If the reached node is the root node and there are some ports to be explored i.e. $(r_L.portentered+1)mod\delta \neq 0$ then $r_L.state$ is updated to $explore$, else $r_L$ and the remaining settled robots are terminated at this node.
            \item If the reached node is other than the root node then $r_L.portentered$ is compared with the value of $r_j.parent$ or $r_j.virtualparent$. If all the ports are explored then the state is changed to $backtrack$ otherwise $r_L$ explores through the incremented value of $r_L.portentered$.
        \end{itemize}
\end{itemize}

The pseudo-code of the algorithm for the last settled robot $r_L$ is given as Algorithm \ref{alg:reach_root}. The Algorithm \ref{alg:reach_root} uses the Algorithm \ref{alg:Begin_Termination} as a subroutine, where Algorithm \ref{alg:Begin_Termination} is the pseudo-code of what the last settled robot $r_L$ does after reaching the root.


\begin{algorithm}
\caption{Help\_Termination(): Algorithm for the robot $r_L$}\label{alg:Begin_Termination}
    {set $\phi = \delta_{u}$, $r_L.state=explore$ and $r_L.portentered = 0$}\\
    {move through $r_L.portentered$}\\
    {$r_L.dist$ = $r_L.dist+1$}\\
    \For{$phase > 0$}
    {
        \If{$r_L.dist = 1$ and $r_L.state = explore$}
        {
        
        {wait for $2\phi$ rounds}\\
            \If{any settled robot visits with $r_j.act\_settled=1$, $r_j.special = 0$, and $r_j.id < $  id of $act\_settled$ robot already present with $r_L$ }
            {
                {$r_L.state$ = $backtrack$}\\
                {update $\phi$= $max\{\phi, \delta(u)\}$}\\
                {move through $r_L.portentered$}\\
                {decrement $r_L.dist$ value to $0$}
            }
            \If{any settled robot visits with $r_j.act\_settled=1$, $r_j.special=0$, but $r_j.id > $ id of $act\_settled$ robot already present with $r_L$ or no $r_j.act\_settled=1$ robot visits}
            {
             
                    {$r_L.state=explore$}\\
                     {update $\phi$ = $max\{\phi , \delta(u)\}$}\\
                        {$r_L.portentered=(r_L.portentered+1)$mod$\delta$}\\
                             \If{$r_L.portentered=r_j.virtualparent$}{
                                    {$r_L.state=backtrack$}\\
                                    {move through $r_L.portentered$}
                                         }
                            \Else{
                                {move through $r_L.portentered$}\\
                                {increment $r_L.dist$ to 2}\\
                                }

                    }

        }
        \If{$r_L.dist = 2$ and $r_L.state = explore$}
        {
        {wait for $2\phi$ rounds}\\
        {update $\phi$= $max\{\phi, \delta(u)\}$}\\
            \If{no settled robot with $r_j.act\_settled=1$ visits}
            {
                {set $r_L.act\_settled=1$}\\
                {terminate $r_L$}
            }
            \If{any settled robot visits with $r_j.act\_settled=1$}
            {
                {set $r_L.state = backtrack$}\\
                {move through $r_L.portentered$}\\
                {decrement the value of $r_L.dist$}
            }
            \If{there is a robot $r_j$ present at the node with $r_j.stage=2$}
            {
                {set $r_L.dist = 0$}\\
                {set $r_L.portentered$ = $(r_L.portentered+1)mod\delta$}\\
                {move through $r_L.portentered$}
            }
        }
        \If{$r_L.state=backtrack$}
        {
        {wait for $2\phi$ rounds}\\
        {set $r_L.portentered = (r_L.portentered+1)mod\delta$}\\
            \If{the $act\_settled$ robot $r_j$ has $r_j.parent=-1$}
            {
                \If{$r_L.portentered = 0$}
                {
                    {set $r_L.act\_settled=1$}\\
                    {terminate $r_L$ and all the remaining settled robots at this node}
                }
                \Else
                {
                    {$r_L.state = explore$}\\
                    {move through $r_L.portentered$}
                }
            }
            
            \Else
            {
                \If{$r_L.portentered = r_j.parent$ or $r_j.virtualparent$}
                {
                {set $r_L.state = backtrack$}\\
                {move through $r_L.portentered$}
                }
                \Else
                {
                {$r_L.state = explore$}\\
                {move through $r_L.portentered$}\\
                {increment $r_L.dist$ by 1}
                }
            }
        }
    }
\end{algorithm}

\begin{algorithm}[h]
\caption{Algorithm for the robot $r_L$ with $r_L.terminate=1$ to reach the root node}\label{alg:reach_root}
\If{$r_L.terminate=1$ and settled at a node other than the root node}{
{set $r_L.stage = 2$, $r_L.state = backtrack$}\\
{move through $r_L.parent$}\\
\While{reached node is not the root node}
{
{wait for $2\phi$ rounds to get the parent port information}\\
{move through the parent pointer of that node}
}
{call Help\_Termination()}
}
\ElseIf{$r_L.terminate=1$ and $r_L$ is settled at the root node }
{call Help\_Termination()}

\end{algorithm}

The robots with $r_j.settled=1$ and $r_j.act\_settled=0/1$ proceeds based on the following cases.
\begin{itemize}
    \item If $r_j$ is a settled robot with $r_j.settled=1$ but $r_j.act\_settled=0$, when it meets any robot $r_i$ with $r_i.stage=2$ then it returns to its original position and waits there. It now becomes aware that the termination stage has started. Thus, it does not increment the value of $r_j.count$.
    \item If $r_j.act\_settled=1$ and $r_j.special=0$ then it continues the back and forth movement through all its neighbors and increments the value of $r_j.count'$ by $1$ as and when it meets $r_L$. When $r_j.count$ becomes equal to $r_j.count'$, it terminates at its original position where it was settled. 
    \item If $r_j.act\_settled=1$ and $r_j.special=1$ then it moves with the robot $r_L$ to the neighboring node through the updated value of $r_L.portentered$ and waits with $r_L$. It increments the value of $r_j.count'$ by $1$. Now two cases arise:
        \begin{itemize}
            \item If no $act\_settled$ robot visits with id lower than $r_j.id$ then $r_j$ sets\\ $r_j.virtualparent=r_L.portentered$
            \item Else move to the original settled position and set $r_j.special=0$
        \end{itemize}
\end{itemize}
The pseudo-code of the algorithm for each settled robot $r_j$ with $r_j.act\_settled=1$ is given in the Algorithm \ref{alg:stage2settled}.

\begin{figure}[]
\begin{subfigure}{.5\textwidth}
  \centering
  \includegraphics[width=.6\linewidth]{Figures/N1.jpg}  
  \caption{}
  \label{fig:sub-first}
\end{subfigure}
\begin{subfigure}{.5\textwidth}
  \centering
  \includegraphics[width=.6\linewidth]{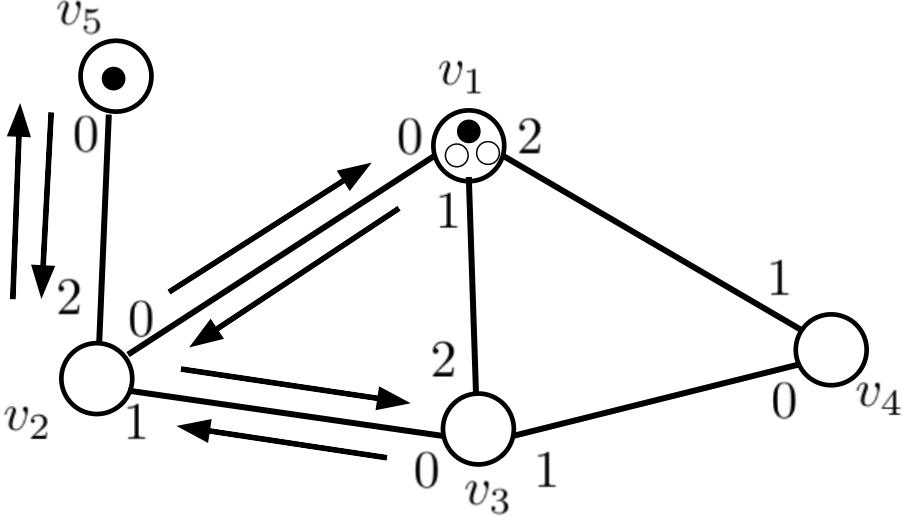}  
  \caption{}
  \label{fig:sub-second}
\end{subfigure}
\vspace{0.7cm}

\begin{subfigure}{.5\textwidth}
  \centering
  \includegraphics[width=.6\linewidth]{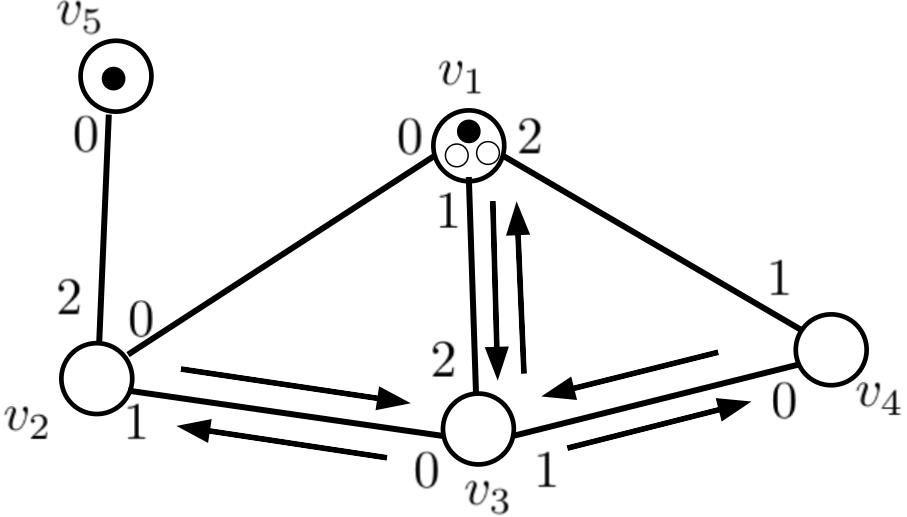}  
  \caption{}
  \label{fig:sub-third}
\end{subfigure}
\begin{subfigure}{.5\textwidth}
  \centering
  \includegraphics[width=.6\linewidth]{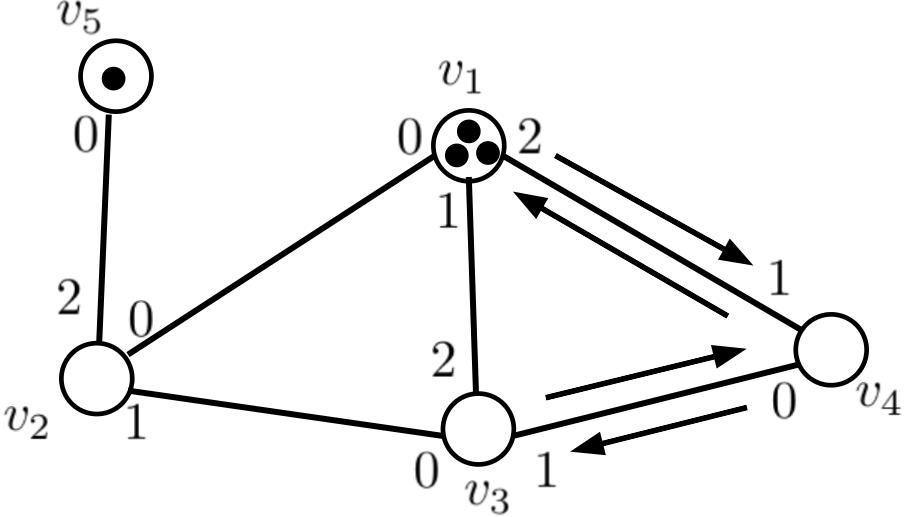}  
  \caption{}
  \label{fig:sub-fourth}
\end{subfigure}
\caption{(a) The initial configuration with four robots at $v_1$.
(b) One robot, say $r$, settles at $v_1$ and updates $r.special=1$. The remaining robots leave through port 0 and the settled robot moves along with the group to the neighboring node $v_2$. The robot $r$ updates the value $r.count=1$. They wait for $2\phi$ rounds and since no settled robot visit the group, they move further to node $v_3$ and update $r_i.dist=2$ while $r$ moves back to $v_1$ after storing 0 as the parent pointer of $v_2$ in $r.virtualparent$ and updates $r.special=0$. The group of unsettled robots wait for $2\phi$ rounds at $v_3$ and due to the visit by the settled robot $r$ they backtrack to the previous node $v_2$ while $r$ updates the value of $r.count=2$ From $v_2$, it explores $v_5$ and since no settled robot visit $v_5$ during the wait of $2\phi$ rounds, one robot, say $r'$ settles there. The group backtracks to $v_2$. The settled robot $r'$ updates $r'.count=1$. The settled robot $r$ increments the value of $r.count$ as well. The group of unsettled robots further backtracks to $v_1$. 
(c) Now the group of unsettled robots as well as $r$, move through port $1$ and explore $v_3$. From $v_3$ it explores $v_2$ while $r$ goes back to $v_1$ after storing $2$ as the parent pointer of $v_3$ in $r.virtualparent$. As and when the settled robots $r$ and $r'$ meets the group of unsettled robots at any node, they increment their value of $count$. During the wait of $2\phi$ rounds, $r$ visits $v_2$ and thus, the group backtracks to the node $v_3$. Then the group explores $v_4$ and due similar reasons, the group backtracks to to $v_3$. After waiting $2\phi$ rounds, it learns the parent pointer from $r$ and backtracks to $v_1$.
(d) Now the group moves through port $2$ to explore $v_4$ and it further moves to explore $v_3$. Due to the visit of settled robot at $v_1$, the group backtracks to $v_4$ and then further backtracks to $v_1$. As $v_1$ is the root and tall the ports of $v_1$ are explored, the unsettled robots settle here at $v_1$. At the end of this stage 1, the value of $r.count=15$ while $r'.count=2$.}
\label{fig:two}
\end{figure}

\begin{algorithm}
\caption{Algorithm for each robot $r_j$ with $r_j.settled=1$ and $r_j.act\_settled=1$ during the termination stage}\label{alg:stage2settled}

\If {$r_j.act\_settled=1$ and $r_j.special=0$}{
{do the back and forth movement through all its ports}\\
    \If{$r_j$ meets $r_L$ with $r_L.terminate=1$}{
    {wait with $r_L$ till it leaves that node}\\
    {set $r_j.count'$ = $r_j.count'+1$}\\
    { move to its original position and then resume the back and forth movement}\\
        \If{$r_j.count == r_j.count'$}{terminate $r_j$}
    }
    }
\ElseIf{$r_j.act\_settled=1$ and $r_j.special=1$}{
{move along with $r_L$ through the updated value of $r_L.portentered$ and wait till $r_L$ leaves}\\
{increment the value of $r_j.count'$}\\
    \If{no $act\_settled$ robot visits with id lower than $r_j.id$}{
        {Set $r_j.virtualparent=r_L.portentered$}\\
        {move to the original position via $r_j.portentered$}\\
        }
    \Else{
    {move to the original position via $r_j.portentered$}
    }
{update $r_j.special=0$}
}
\end{algorithm}
The pseudo-code of the Distance-2-Dispersion with Termination is given in the Algorithm \ref{alg:mainalgorithm}. Figure \ref{fig:two} shows the run of stage 1 of our algorithm.
\begin{algorithm}
\caption{D-2-D\_with\_Termination}\label{alg:mainalgorithm}
\If{$r_i.settle=0$}
{
call Algorithm \ref{alg:cap}
}
\ElseIf{$r_i.settle=1$ and $r_i.act\_settled=0$}
{
call Algorithm \ref{alg:settle}
}
\ElseIf{$r_i.settle=1$ and $r_i.terminate=1$}
{
call Algorithm \ref{alg:reach_root}
}
\ElseIf{$r_i.settle=1$ and $r_i.act\_settled=1$}
{
call Algorithm \ref{alg:stage2settled}
}
\end{algorithm}




\subsection{Analysis of the Algorithm}
    \begin{definition}{ Tree Edge:}\label{def:1}
       An edge $(u,v)$ is said to be a tree edge if the group of unsettled robots in stage 1 reaches $v$ through $(u,v)$ such that either the settled robot at node $u$ (if exists) stores the parent pointer of the node $v$ or the minimum id robot among the group of unsettled robots settles at $v$. 
    \end{definition}
\begin{remark}
       A settled robot $r_i$ in stage 1 stores the parent pointer for its adjacent node $u$ at some round $t$ only if $r_i.special=1$, $r_j.dist=1$ for any  unsettled robot $r_j$  at $u$ and no settled robot with id lower than $r_i.id$ visits $r_j$ during the $2\phi$ rounds i.e., during the waiting period. 
\end{remark}

    \begin{theorem}\label{th:nonadjacent}
    By the end of the Algorithm \ref{alg:mainalgorithm}, there are no two robots that are settled at adjacent nodes.
    \end{theorem}

    \begin{proof}
    The Algorithm \ref{alg:settle} for the settled robots in stage 1 guarantees that the settled robots show their presence by back-and-forth movement to their one-hop neighbors. Thus, when the group of unsettled robots visits a node and wait for $2\phi$ rounds, then they encounter the settled robot, if present, in its one-hop neighbor. And according to the Algorithm \ref{alg:cap} for unsettled robots in stage 1, no unsettled robot settles if some settled robot meets the unsettled robot in some node. Also according to Algorithm \ref{alg:stage2settled}, the settled robots in stage 2 settle at nodes where they get settled in stage 1. This guarantees that no two adjacent nodes are occupied by the robots.
    \end{proof}
    
     \begin{lemma}\label{lem:multiatroot}
     Multiple robots can settle only at the root. 
    \end{lemma}
    \begin{proof}
    When the robots complete the traversal of the graph and do not find any node to settle satisfying the conditions of the D-2-D problem, they finally reach the root of the graph to continue Algorithm \ref{alg:cap} and traverse through the root of the graph. The robots can easily recognize the root as the parent pointer of the root is $-1$. After following the algorithm from the root and subsequently exploring through all outgoing edges, robots backtrack to root only if they don't find nodes to settle. In this case, they settle at the root. Thus, algorithm \ref{alg:cap} leads multiple robots to settle only at the root.
  \end{proof}
    
      \begin{lemma}\label{lem:conn}
    If multiple robots settle at the root in stage 1, then it is guaranteed that each node is visited by a group of unsettled robots at least once.
    \end{lemma}
    \begin{proof}
    Let us suppose there is a node $u$ which is not visited at all but at least one of its one-hop neighbors, say $v$, is visited. This implies, that every time the group reached $v$, either it backtracked from $v$ or it explored all the ports except the port joining $v$ with $u$. The latter case is not possible as the Algorithm \ref{alg:cap} increments the value of $portentered$ unless its value is equal to the value of the parent pointer. 
    
    Now without loss of generality let us consider the case when $v$ is visited by the group from the node that contains the lowest id settled robot among the ids of the settled robots at one-hop neighbors of $v$. The Algorithm \ref{alg:cap} ensures that the group of unsettled robots backtracks from $v$ only when all the ports of $v$ are explored. And hence $u$ must be explored and this is a contradiction to the existence of such a node $u$. 
  \end{proof}

    \begin{theorem}\label{th:novacant}
    By the end of the Algorithm \ref{alg:mainalgorithm}, multiple robots settled at the root implies no vacant node left such that none of its neighbors contains a settled robot.
    \end{theorem}
\begin{proof}
    Let us suppose there is a vacant node $u$ in the graph such that no settled robot is present in any of its one-hop neighbors in the stage 1. Lemma 1 proves that node $u$ is visited at least once. According to our algorithm for unsettled robots in stage 1, i.e. Algorithm \ref{alg:cap}, when the group visited $u$, each of the robots $r_j$ in the group must set $r_j.dist=2$. During the waiting period, there were no settled robots in the neighbors of $u$ to visit $u$. Hence the minimum id robot must have settled there. This contradicts the presence of such a node in the graph.
   \end{proof} 

        \begin{observ}\label{obs:maxis}
     If multiple robots settle in the root, it follows from Theorem \ref{th:nonadjacent} and Theorem \ref{th:novacant} that the nodes with settled robots form a maximal independent set.
      \end{observ}
    
    \begin{theorem}
    D-2-D with termination can be run by the robots with $O(\log \Delta)$ additional memory.
    \end{theorem}
    \begin{proof}
    The variables $r_i.state$, $r_i.stage$, $r_i.settled$, $r_i.act\_settled$ and $r_i.special$ requires 1 bit of memory while $r_i.dist$ requires 2 bits of memory. The variables $r_i.parent$, $r_i.portentered$ and $r_i.virtualparent$ requires $O(\log \Delta)$ bits of memory. The settled robot at a node $v$ with $\delta(v)\leq\Delta$ can meet the group of unsettled robots at at most $(\Delta+1)$ nodes including node $v$ and there can be at most $O(\Delta^2)$ associated edges with these nodes. Since the group of unsettled robots visits any edge at most 4 times, the variable $r_i.count$ can take maximum value that is in $O(\Delta^2)$. Similarly, in stage 2 the $act\_settled$ robot at a node $v$ with $\delta(v)\leq\Delta$ can meet $r_L$ at $(\Delta+1)$ nodes and thus, $r_i.count'$ can take maximum value that is in $O(\Delta^2)$. Therefore, $O(log \Delta)$ is the amount of memory needed by the robots to store the information relating to these variables. As a result, each robot only needs $O(\log \Delta)$ bits of additional memory to run the algorithm.
    \end{proof}
    
    \begin{lemma}\label{lem:special}
    When the group of unsettled robots in stage 1 are in $explore$ state and $r_i.dist=1$ then there is exactly one settled robot present along with the group which has $r_i.special=1$.
   \end{lemma} 
   \begin{proof} It is easy to observe this from the description of $r_i.special$ variable of a settled robot as mentioned in table \ref{table:symbols}. As no node except the root contains multiple settled robots. Also the root contains multiple robots only when no robots are left to settle, i.e. no robot is in the explore state in stage 1 anymore. Hence, the statement follows.  
   \end{proof}
   
        \begin{lemma}\label{lem:tree}
    Every tree edge in stage 1 is traversed exactly twice by the group of unsettled robots.  
    \end{lemma}
    \begin{proof}
    Without loss of generality, according to Definition \ref{def:1}, let $u$ has a settled robot. the tree edge $(u,v)$ has either a settled robot at $v$, or a settled robot at $u$ that stores the parent pointer for node $v$ during the exploration of edge $(u,v)$. This ensures that $v$ is visited for the first time as we have its parent pointer stored. Thus, the edge $(u,v)$ is traversed twice once in the $explore$ state and the next in the $backtrack$ state. As mentioned in Algorithm \ref{alg:settle}, the parent pointer of node $v$ is saved by robot $r_u$ settled at node $u$ only when no robot visits $v$ with id $< r_u.id$. Hence the robots do not backtrack from $v$ with the objective of exploring node $v$ from another node with lower id robot settled on it. This proves the edge $(u,v)$ is traversed exactly two times.   
    \end{proof}
    
    \begin{lemma}\label{lem:nontree}
    Every non-tree edge is traversed at most four times by the group of unsettled robots.
    \end{lemma}
     \begin{proof}
    Let $(u,v)$ be a non-tree edge. According to Definition \ref{def:1}, the robots backtrack from node $v$ and the parent pointer for $v$ is not yet stored. Till this round, the edge $(u,v)$ has been traversed twice. The robots reach $v$ from the smallest id settled robot in its neighborhood to explore $v$ later. At that time edge $(v,u)$ is traversed again. Hence, every non-tree edge is traversed at most four times.
    \end{proof}
    
    \begin{lemma}\label{lem:spanning}
   The graph induced by the tree edges is connected and cycle free.  
    \end{lemma}
    \begin{proof}
    Consider a rooted configuration on a graph $G$ with root $u$ such that degree of $u$ is at least 2 and also $k\ge 1$. First we show that the tree edges form a connected component. It is easy to see that first two tree edges form a connected component. Let $e_1, e_2,..., e_h$ be the first $h$ tree edges and they form a connected component. Let there be still a group of unsettled robots that is doing the traversal. Let $e_h=\overline{uw}$ for some nodes $u$, $w$ such that the tree edge was formed when the group of unsettled robot visited $w$ from $u$. If $\overline{ww'}$ becomes a tree edge for some neighbor $w'$ of $w$ then we are done, i.e., the $(h+1)$th tree edge also remains in the same connected component. Else, if no more associated edge of $w$ becomes a tree edge, the group backtracks from $w$ via a tree edge and reaches a new node $v$, say. Again, either an adjacent edge of $v$ becomes a new tree edge (in which case we are done), or it backtracks  through another tree edge. And in this way the group continues to stay on a path consisting of tree edges until if finds a new tree edge, or it completes the exploration and all the robots of the group settles at the root. Whatever be the case, the tree edges form a connected component. 
    
    Now we prove that the induced graph is cycle free. Let us assume on contrary that there is a cycle consisting of the tree edges. Let $u_1$, $u_2$, ..., $u_p$, $u_1$ be the cycle consisting of the tree edges. W.l.o.g., assume that $\overline{u_iu_{i+1}}$ be the last tree edge due to which cycle is formed and group of unsettled robots moved from $u_i$ to $u_{i+1}$. Now we have two cases: either there is a settled robot at $u_{i+1}$ or there is no settled robot at $u_{i+1}$. In case there is a settled robot at $u_{i+1}$, then the group of unsettled robots should have done a backtrack from $u_{i+1}$ to $u_i$ and hence $\overline{u_iu_{i+1}}$ can not be a tree edge. This is a contradiction to our assumption. So, let us assume there is no settled robot at $u_{i+1}$. Definition \ref{def:1} implies there will be settled robots both at $u_i$ and $u_{i+2}$. Now, $u_{i+1}$ is at one hop distance from these two settled robots and the exploration is being done from $u_i$ to $u_{i+1}$. Either of the two settled robots at $u_i$ and $u_{i+2}$ have smaller id. If the settled robot at $u_{i+2}$ has smaller id then the robots will backtrack from $u_{i+1}$ to $u_i$ and thus $\overline{u_iu_{i+1}}$ will not be a tree edge. However, if the settled robot at $u_i$ has smaller id then while exploring the node $u_{i+2}$ and traversing from $u_{i+2}$ to $u_{i+1}$, the group of unsettled robots must have backtracked due to presence of a smaller id settled robot at $u_i$ thus forming $\overline{u_{i+2}u_{i+1}}$ as the non tree edge. Thus, we see that $\overline{u_iu_{i+1}}$ and $\overline{u_{i+2}u_{i+1}}$ cannot be tree edges simultaneously. Hence, our assumption of the presence of a cycle consisting of all the tree edges is wrong and the graph induced by the tree edges is connected and cycle free.
    \end{proof}
    
 \begin{lemma}
    By the time stage 2 finishes, each robot terminates.\label{lem:7}
 \end{lemma}
 \begin{proof}
    Since the robot $r_L$ with $r_L.terminate=1$ replicates the group of unsettled robots in stage 1 and all the robots with $r_L.act\_settled=1$ replicates the settled robots in stage 1, so, the number of times each settled robot meets with the group of unsettled robots in stage 1 is same as the number of times each $act\_settled$ robot meets with $r_L$.
    As mentioned in section \ref{sec:3.1}, the stage 2 is replay of stage 1. So the correctness of stage 1 implies the correctness of stage 2. And hence for each settled robot $r_i$ except $r_L$, $r_i.count=r_i.count,$ and terminates. Finally, $r_L$ settles at the node where it settled at the end of stage 1 and terminates.
    \end{proof}
 
\begin{theorem}
The Algorithm \ref{alg:mainalgorithm} achieves D-2-D with termination in $2\Delta(8m-3n+3)$ rounds on arbitrary rooted graphs.
\end{theorem}
\begin{proof} 
     It is clear from Lemma \ref{lem:tree} and Lemma \ref{lem:nontree} that every edge is traversed at most 4 times except the tree edges. Also from Lemma \ref{lem:spanning}, there can be at most $(n-1)$ tree edges. So the total number of edge traversal is no more than $4(m-(n-1))+2(n-1)=4m-2n+2$. After each edge traversal, the robots wait for $2\phi$ rounds and $\phi \leq \Delta$. So at most $2\Delta (4m-2n+2)$ rounds are required for all the robots to settle. Thus Stage 1 is completed within $2\Delta (4m-2n+2)$ many rounds.
     After the last robot settles, it may take at most $2\Delta(n-1)$ rounds to reach the root node in the worst-case. Now, the remaining part of stage 2 is replica of the stage 1 of our algorithm. Thus, it takes $2\Delta(8m-3n+3)$ many rounds in order to achieve D-2-D with termination
     \end{proof}
     
\section{Lower Bound}\label{sec:lowerbound}
In this section we discuss the lower bound on number of rounds of D-2-D problem considering robots do not have more than $O(\log\Delta)$ additional memory. We start by defining view of a node to a robot.
\begin{definition}\label{def:view}
\noindent\textbf{View:} View of a node $v$ to a robot is the information of whether there is a settled robot at any of its one hop neighbor or not, including $v$.  
\end{definition}

Next we prove the theorem by constructing a class of graphs. The idea is that, each graph in the class is a regular graph of degree $n-1$ and has $2n$ nodes. We start with two robots, one of which settles first and the other looks for a node to settle. The graphs are such that, unless the unsettled robot reaches two particular nodes, it will not be able to differentiate the graph with a clique. So, before reaching one of those nodes, if it decides to settle, that will lead to a wrong solution. We show that, with limited memory, finding one of those nodes requires at least 
$\Omega(m\Delta)$ rounds.

\begin{theorem}
The lower bound on number of rounds of D-2-D problem on arbitrary graphs is $\Omega(m\Delta)$ considering robots have no more than $O(\log\Delta)$ additional memory.
\end{theorem}
\begin{proof}
We will prove this using a class of graphs where we show that there will be at least one graph for which the robots require at least  $\frac{\Delta m}{12}$ many rounds to complete D-2-D. Let us consider two cliques of $n$ vertices but with one edge missing from each of them. Let $v_1$, $v_2$, ..., $v_n$ be the vertices of the first clique $Q_1$ and $u_1$, $u_2$, ..., $u_n$ be the vertices of the second clique $Q_2$. Let $\overline{v_1v_2}$ be the missing edge from the first clique and $\overline{u_1u_2}$ be missing from the second clique. We join $v_1$ with $u_1$ and $v_2$ with $u_2$. Now, the graph $G$ has $2n$ nodes with $\Delta=n-1$. Considering all possible different port-numbering of this graph gives us a graph class $\mathscr{G}$ which has cardinality equal to $[(n-1)!]^{2n}$. Let two robots $r_1$ and $r_2$ are initially present at $v_j$ where $j\neq{1,2}$. Let us assume that there exists an algorithm $\mathscr{A}$ which solves D-2-D in time less than $\frac{m\Delta}{12}$. Let $r_1$ settles first and at node $w$. We can claim that there will be at least $\frac{|\mathscr{G}|}{2}$ graphs where, $w \notin \{ v_1, v_2, u_1, u_2 \}$. W.l.o.g. let $w=v_i$ be some vertex of $Q_1$. Let us denote $\frac{|\mathscr{G}|}{2}$ by $N$.

As the robots have $O(\log \Delta)$ memory, they can remember only a constant many port numbers at a time. We provide $r_2$ more power by letting it know that there is a node to settle within two hop distance of $v_i$. The robot $r_2$ aims to explore all the $\Delta(\Delta-1)$ many two hop neighbors. There are enough graphs(in particular, $\frac{N}{4}$) wherein the robot $r_2$ needs to explore at least $\frac{\Delta(\Delta-1)}{2}$ many vertices before exploring $u_1$ or $u_2$. Unless it reaches $u_1$ or $u_2$ and has the view, $r_2$ can not distinguish any graph of our graph class from a clique of $n$ nodes. 

Let the sequence in which the nodes are explored is as follows \{$v_{i_1}$, $v_{i_2}$,..., $v_{i_{\frac{\Delta(\Delta-1)}{2}}}$\}. When $r_2$ reaches $v_{i_1}$, it needs to know the view of the graph. 
If $v_{i_1}$ is reached from $v_i$ directly, then getting the view takes only one round as $r_2$ understands it is one hop away from $v_i$. Else, if $v_{i_1}$ is not reached directly from $v_i$, then it is easy to see that, in at least half of the graphs, $r_2$ needs at least $\frac{\Delta}{2}$ rounds to get the view.
So, there exists enough instances(in particular at least $\frac{N}{4.2}$) where $r_2$ requires $\frac{\Delta}{2}$ rounds to find the view. Similarly, after reaching $v_{i_2}$ there exists at least $\frac{N}{4.2^2}$ many graphs where $\frac{\Delta}{2}$ many rounds will be required to find the view of that node. In similar fashion, at $v_{i_{\frac{\Delta(\Delta-1)}{2}}}$ there exists at least $\frac{N}{4.2^{\frac{\Delta(\Delta-1)}{2}}}$ many graphs. Now $\frac{N}{4.2^{\frac{\Delta(\Delta-1)}{2}}}$ is a function of $n$ and the value becomes more than 1 for all $n\geq3$. 

Hence, there is at least one graph where robot $r_2$ needs to spend at least $\frac{\Delta(\Delta-1)}{2}.\frac{\Delta}{2}$ rounds to settle. For $n\geq3$, $\Delta \geq \frac{M}{3}$ where $M=2n$. Thus, $\frac{\Delta(\Delta-1)}{2}.\frac{\Delta}{2} \geq \frac{M}{3}.\frac{(\Delta-1)}{2}.\frac{\Delta}{2}=m.\frac{\Delta-1}{6} \geq \frac{m\Delta}{12}$. This proves there is at least one such instance in the class $\mathscr{G}$ where the robot $r_2$ requires $\frac{m\Delta}{12}$ many rounds to complete D-2-D, else both $r_1$ and $r_2$ settles either on $Q_1$ or on $Q_2$ and this leads to wrong D-2-D. This completes the proof.
\end{proof}

\section{Conclusion and Future Work}\label{Conclusion}
We propose a variant of the dispersion problem and provide an algorithm that solves it for the rooted initial configuration with $O(\log\Delta)$ additional memory per robot and in $2\Delta(8m-3n+3)$ synchronous rounds. We also provide a $\Omega(m\Delta)$ lower bound of the problem on number of rounds. In some cases, we guarantee forming a maximal independent set by the robots which can be of independent interest. It will be interesting to see how to solve the problem for arbitrary initial configuration of the robots. 

\bibliographystyle{plainurl}
\bibliography{Bib.bib}

\begin{thebibliography}{10}

\bibitem{AgarwallaAMKS18}
Ankush Agarwalla, John Augustine, William K.~Moses Jr., Sankar~Madhav K., and
  Arvind~Krishna Sridhar.
\newblock Deterministic dispersion of mobile robots in dynamic rings.
\newblock In {\em {ICDCN}}, pages 19:1--19:4. {ACM}, 2018.

\bibitem{Augustine18}
John Augustine and William K.~Moses Jr.
\newblock Dispersion of mobile robots: {A} study of memory-time trade-offs.
\newblock In {\em ICDCN}, pages 1:1--1:10, 2018.

\bibitem{BarriereFBS09}
Lali Barri{\`{e}}re, P~Flocchini, E~M Barrameda, and N~Santoro.
\newblock Uniform scattering of autonomous mobile robots in a grid.
\newblock In {\em {IPDPS}}, pages 1--8, 2009.

\bibitem{DasBS21}
Archak Das, Kaustav Bose, and Buddhadeb Sau.
\newblock Memory optimal dispersion by anonymous mobile robots.
\newblock In {\em CALDAM}, pages 426--439, 2021.

\bibitem{ShantanuDas19}
Shantanu Das.
\newblock Graph explorations with mobile agents.
\newblock In {\em Distributed Computing by Mobile Entities, Current Research in
  Moving and Computing}, pages 403--422. Springer, 2019.

\bibitem{DereniowskiDKPU15}
D.~Dereniowski, Y.~Disser, Adrian Kosowski, D.~Pajak, and P.~Uznanski.
\newblock Fast collaborative graph exploration.
\newblock {\em Inf. Comput.}, 243:37--49, 2015.

\bibitem{ElorB11}
Yotam Elor and Alfred~M. Bruckstein.
\newblock Uniform multi-agent deployment on a ring.
\newblock {\em Theor. Comput. Sci.}, 412(8-10):783--795, 2011.

\bibitem{Barun}
Barun Gorain, Partha~Sarathi Mandal, Kaushik Mondal, and Supantha Pandit.
\newblock Collaborative dispersion by silent robots.
\newblock {\em CoRR}, abs/2202.05710, 2022.

\bibitem{KshemkalyaniMS19}
A~D. Kshemkalyani, A~R Molla, and G~Sharma.
\newblock Fast dispersion of mobile robots on arbitrary graphs.
\newblock In {\em ALGOSENSORS}, pages 23--40, 2019.

\bibitem{kshemkalyaniMS20}
A.~D. Kshemkalyani, A.~R. Molla, and G.~Sharma.
\newblock Dispersion of mobile robots on grids.
\newblock In {\em {WALCOM}}, volume 12049, pages 183--197. Springer, 2020.

\bibitem{KshemkalyaniA19}
Ajay~D. Kshemkalyani and Faizan Ali.
\newblock Efficient dispersion of mobile robots on graphs.
\newblock In {\em ICDCN}, pages 218--227, 2019.

\bibitem{KshemkalyaniMS22}
Ajay~D. Kshemkalyani, Anisur~Rahaman Molla, and Gokarna Sharma.
\newblock Dispersion of mobile robots using global communication.
\newblock {\em J. Parallel Distributed Comput.}, 161:100--117, 2022.

\bibitem{KshemkalyaniS21}
Ajay~D. Kshemkalyani and Gokarna Sharma.
\newblock Near-optimal dispersion on arbitrary anonymous graphs.
\newblock In {\em OPODIS}, pages 8:1--8:19, 2021.

\bibitem{MollaM19}
Anisur~Rahaman Molla and William K.~Moses Jr.
\newblock Dispersion of mobile robots: The power of randomness.
\newblock In {\em {TAMC}}, volume 11436, pages 481--500, 2019.

\bibitem{Ani}
Anisur~Rahaman Molla, Kaushik Mondal, and William K.~Moses Jr.
\newblock Efficient dispersion on an anonymous ring in the presence of weak
  byzantine robots.
\newblock In {\em {ALGOSENSORS}}, volume 12503, pages 154--169. Springer, 2020.

\bibitem{Molla0M21}
Anisur~Rahaman Molla, Kaushik Mondal, and William K.~Moses Jr.
\newblock Byzantine dispersion on graphs.
\newblock In {\em IPDPS}, pages 942--951. {IEEE}, 2021.

\bibitem{MollaMM21}
Anisur~Rahaman Molla, Kaushik Mondal, and William K.~Moses Jr.
\newblock Optimal dispersion on an anonymous ring in the presence of weak
  byzantine robots.
\newblock {\em Theor. Comput. Sci.}, 887:111--121, 2021.

\bibitem{sai}
S.~Vamshi Samala, S.~Pramanick, D.~Pattanayak, and P.~S. Mandal.
\newblock Filling {MIS} vertices by myopic luminous robots.
\newblock {\em CoRR}, abs/2107.04885, 2021.

\bibitem{ShibataMOKM16}
Masahiro Shibata, Toshiya Mega, Fukuhito Ooshita, Hirotsugu Kakugawa, and
  Toshimitsu Masuzawa.
\newblock Uniform deployment of mobile agents in asynchronous rings.
\newblock {\em J. Parallel Distributed Comput.}, 119:92--106, 2018.

\bibitem{ShintakuSKM20}
Takahiro Shintaku, Yuichi Sudo, Hirotsugu Kakugawa, and Toshimitsu Masuzawa.
\newblock Efficient dispersion of mobile agents without global knowledge.
\newblock In {\em SSS}, volume 12514, pages 280--294, 2020.

\end{thebibliography}

\end{document}